%% file: main.tex
\documentclass[11pt]{article}

\usepackage{fullpage}
\usepackage[round]{natbib}
\usepackage{dsfont}
\usepackage{latexsym}
\usepackage{amsmath}
\usepackage{amsthm}
\usepackage{amssymb}
\usepackage{amsfonts}
\usepackage{mathrsfs}
\usepackage{aliascnt}
\usepackage{graphicx}
\usepackage{enumerate}
\usepackage[pdfpagelabels,pdfpagemode=None]{hyperref}
\usepackage{algorithmic}
\usepackage{algorithm}
\usepackage{color}

\newtheorem{theorem}{Theorem}%

\newaliascnt{lemma}{theorem}
\newtheorem{lemma}[lemma]{Lemma}%
\aliascntresetthe{lemma}

\newtheorem{corollary}{Corollary}%

\theoremstyle{definition}
\newtheorem{definition}{Definition}

\newtheorem{remark}{Remark}

\newcommand{\bjl}[1]{}
\newcommand{\hf}[1]{}
\newcommand{\ps}[1]{}

\input{notation}

\title{Randomization beats Second Price as a Prior-Independent Auction}

\author{Hu Fu%
\thanks{Microsoft Research New England, hufu@microsoft.com} \and %
 Nicole Immorlica%
\thanks{Microsoft Research New England, nicimm@gmail.com} \and %
  Brendan Lucier%
\thanks{Microsoft Research New England, brlucier@microsoft.com} \and
  Philipp Strack%
 \thanks{University of California Berkeley, pstrack@berkeley.edu}}

\begin{document}

\maketitle

\begin{abstract}
\input{abstract}
\end{abstract}

\section{Introduction}
\label{sec:intro}
\input{intro}

\section{Preliminaries}
\label{SEC:PRELIM}
\input{prelim}

%
%
%
%

\section{Inflated Second Price Auction}
\label{SEC:SPAN}
\input{span}


\section{Revenue Maximization with a Single Sample}
\label{SEC:SINGLESAMPLE}
\input{singlesample}

\bibliographystyle{apalike}
\bibliography{bibs}

\appendix

\section{Missing Argument from \autoref{SEC:PRELIM}}
\label{sec:prelim-app}
\input{prelimapp}

\section{Missing Proofs from \autoref{SEC:SPAN}}
\label{sec:span-app}
\input{spanapp}

\section{Missing Proofs from \autoref{SEC:SINGLESAMPLE}}
\label{sec:singlesample-app}
\input{singlesampleapp}

\end{document}

%% file: notation.tex
\newcommand{\AutoAdjust}[3]{\mathchoice{ \left #1 #2  \right #3}{#1 #2 #3}{#1 #2 #3}{#1 #2 #3} }
\newcommand{\Xcomment}[1]{{}}

\newcommand{\InBrackets}[1]{\AutoAdjust{[}{#1}{]}}
\newcommand{\Ex}[2][]{\operatorname{\mathbf E}_{#1}\InBrackets{#2}}

\newcommand{\dd}{\mathrm{d}}  

\newcommand{\eps}{\epsilon}

\newcommand{\noaccents}[1]{#1}

\newcommand{\newagentvar}[3][\noaccents]{%
\expandafter\newcommand\expandafter{\csname #2\endcsname}{#1{#3}}%
\expandafter\newcommand\expandafter{\csname #2s\endcsname}{#1{\boldsymbol{#3}}}%
\expandafter\newcommand\expandafter{\csname #2smi\endcsname}[1][i]{#1{\boldsymbol{#3}}_{-##1}}%
\expandafter\newcommand\expandafter{\csname #2i\endcsname}[1][i]{#1{#3}_{##1}}%
\expandafter\newcommand\expandafter{\csname #2ith\endcsname}[1][i]{#1{#3}_{(##1)}}%
}

\newcommand{\blowfac}{\delta}
\newcommand{\blowprob}{\epsilon}
\newcommand{\dist}{F}
\newcommand{\dens}{f}
\newagentvar{val}{v}
\newagentvar{quant}{q}
\newcommand{\monresquant}{\quant_*}
\newcommand{\monres}{\val_*}
\newcommand{\rev}{R}

\newcommand{\quantquestion}{\tilde \quant}

\newcommand{\sample}{s}
\newcommand{\shadefac}{\rho}
\newcommand{\shadeprob}{\zeta}

\newcommand{\revshade}{\beta}
\newcommand{\revshadedown}{\gamma}
\newcommand{\quantshadedown}{\quant_{\revshadedown}}
\newcommand{\quantshade}{\varphi}

%% file: abstract.tex
Designing revenue optimal auctions for selling an item to $n$ symmetric bidders is a fundamental problem in mechanism design.  Myerson (1981) shows that the second price auction with an appropriate reserve price is optimal when bidders' values are drawn i.i.d.\@ from a known regular distribution.  A cornerstone in the prior-independent revenue maximization literature is a result by Bulow and Klemperer (1996) showing that the second price auction without a reserve achieves $(n-1)/n$ of the optimal revenue in the worst case.

We construct a randomized mechanism that strictly outperforms the second price auction in this setting.  Our mechanism inflates the second highest bid with a probability that varies with $n$.  For two bidders we improve the performance guarantee from $0.5$ to $0.512$ of the optimal revenue.  We also resolve a question in the design of revenue optimal mechanisms that have access to a single sample from an unknown distribution.  We show that a randomized mechanism strictly outperforms all deterministic mechanisms in terms of worst case guarantee.


%% file: intro.tex
Designing revenue optimal auctions to sell an item to symmetric bidders is one of the most fundamental problems in optimal mechanism design.
%
%
This problem was solved by \citet{Mye81}: %
if the buyers' valuations for the item are independently drawn from a distribution which is known to the seller, the revenue-maximizing mechanism is a second price auction (SPA) with a reserve price.
%
%
Notably, to determine the reserve price and implement the optimal auction the seller needs to know the valuation distribution.
Since setting the reserve price too high might leave the seller with no revenue, the design of auctions with good revenue guarantee when the prior information is unavailable, or very costly to learn, is of high economic relevance. 
%
%
%
%

%
%
A first answer to this question was given by \citet{bu1996}, who show that for regular%
\footnote{A distribution is regular if the virtual valuation is non-decreasing. This assumption, which is standard in the mechanism design literature, will be discussed in Section~\ref{SEC:PRELIM}. 
\hf{We might be making too strong a statement here.} \bjl{I softened the statement slightly --- is it better?}}
 distributions a second price auction with a reserve price of zero guarantees the seller at least a $1-1/n$ fraction of the optimal revenue, where $n$ is the number of bidders.  In particular, this yields at least half of the optimal revenue whenever there are at least two bidders in the auction.
%
%
Note that this auction is \emph{prior-independent}, in the sense that it does not depend on the prior distribution from which valuations are drawn.
%
%
Moreover, it is clear that this is the best possible way to set a deterministic reserve price:
any reserve greater than zero yields zero revenue when all bidders valuations fall below it, so no deterministic positive reserve can guarantee the seller a higher share of the optimal revenue.
In other words, the fixed reserve price that guarantees the seller the highest share of the optimal revenue, against all regular distributions, equals zero.
In fact, with the characterization developed by Myerson, it is not hard to show that \emph{no} symmetric \emph{deterministic} auction can improve upon Bulow and Klemperer's guarantee of half of the optimal revenue (for two bidders).
%
%

This paper gives the first mechanism that outperforms the second price auction in \citet{bu1996}'s setting, by way of making use of randomization.
%
%
This mechanism, which we call \textit{bid inflation}, works as follows:
with a fixed probability $1-\epsilon < 1$, the mechanism runs a SPA with a reserve price of zero; i.e., each bidder gets the object if and only if his valuation is above all other bidders' valuations.
With the remaining probability $\epsilon > 0$, the mechanism allocates the object to the bidder with the highest valuation if and only if his valuation is greater than $(1+\delta) > 1$ times the valuation of any other bidder, otherwise the object is unallocated.
%
%

The idea behind \textit{bid inflation} is based upon a refined analysis of the result by Bulow and Klemperer (1996).
%
%
Our analysis distinguishes two types of distributions: those for which the optimal reserve price is high, 
and those for which it is low.  
A reserve price is considered high (low) if the probability that the valuation exceeds it is low (high).
We first show that whenever the optimal reserve price is sufficiently low, the second price auction without reserve obtains strictly more than a $(1-1/n)$ fraction of the optimal revenue.
\hf{How about this for the second scenario:} \bjl{Love it!  Using the new text.} 
This suggests that one should design a mechanism which works well for distributions with high reserve prices.  An analysis \`a la Myerson suggests that one should try not to allocate to bids below the optimal reserve price and only allocate to the bids above it (even without knowing the optimal reserve).  Whereas inflating the second highest bid will correctly keep the item unsold when the highest bid is not high enough, this will also inevitably lose revenue when the highest bid is high but the second highest bid is not far below.  The key to our analysis is to show that for distributions with high reserves, the loss from the latter scenario does not offset the gain from the former one.  This in turn relies on a technical lemma which controls the \emph{quantile}, i.e., the relative standing in a distribution, as a valuation changes in a regular distribution.  With this analysis for the inflated second price auction in hand, one can then show via a continuity argument that a probabilistic mixture of the inflated auction and the second price auction without a reserve price guarantees a better approximation for all distributions.

%
%
While the main point of this paper is to demonstrate that
there exist randomized mechanisms that outperform the second-price auction in a worst case analysis over regular distributions, we also quantify the improvement in the two bidder case.
For two bidders, a second price auction with zero reserve price guarantees the seller half of the optimal revenue.
We show that the bid inflation mechanism, where the higher valuation bidder receives the object with probability $85\%$ if his valuation is below twice the other bidder's valuation, guarantees a revenue of at least $51.2\%$ of the optimal revenue.
Thus, for two players the bid inflation mechanism generates a $2.4\%$ higher revenue guarantee in the worst case scenario.


Finally, we also consider a variant of the prior-independent setting, in which the seller gains a small amount of information about the prior distribution in the form of a sample.  This sampling model has gained recent attention in the computer science literature.  In a spirit similar to the result of Bulow and Klemperer, it is known that running a second-price auction, with a reserve set equal to the sampled value, guarantees at least $1 - \frac 1 {n+1}$ fraction of the optimal revenue for arbitrary regular distributions; when there is only a single bidder, posting the sampled value as a take-it-or-leave-it price gives therefore half of the optimal revenue \citep{dhangwatnotai2014revenue}.  It is also known that no deterministic pricing method can improve upon this guarantee for a single bidder \citep{huang2014making}.  However, we show how to construct a \emph{randomized} mechanism which improves upon this revenue guarentee and obtains more than half of the optimal revenue for a single bidder.  The construction is of a flavour similar to the bid-inflation mechanism: it offers the sampled value as a take-it-or-leave-it price, but sometimes either inflates the value or shades it down before using it.

\citet{dhangwatnotai2014revenue} uses the single sample technique to construct a prior-independent auction for settings where bidders' values are independent draws from non-identical regular distributions but each distribution gives rise to at least two bidders' values.  Our improvement for the single bidder single sample scenario immediately implies a better prior-independent auction in this setting.

\subsection*{Related Literature}
As the \textit{bid inflation mechanism} is independent of the distribution of valuations it follows the doctrine propagated by \citet{Wilson1989} that economic mechanisms should not rely on precise details of the environment. 
%
%
Thereby, this paper falls in a recent literature in economics and computer science that aims at constructing mechanisms that work well under different distributions of valuations \citep[e.g.][]{bergemann2008pricing, bergemann2011should, dhangwatnotai2014revenue, RTY12, DHKN11,FHH13, FHHK14}.
The second part of our paper is also closely related to a recent literature that analyzes the optimal use of a single sample in this context \citep{dhangwatnotai2014revenue, huang2014making}.  \citet{huang2014making} showed that posting the sampled value as a take-it-or-leave-it price, which guarantees half of optimal revenue, is the optimal deterministic mechanism.  
We show in \autoref{SEC:SINGLESAMPLE} that randomly inflating and shading the sample value can strictly improve the guarantee.
In contrast to studying the use of a single sample, \citet{CR14} give the asymptotic number of samples needed for a given approximation factor, in the presence of asymmetric bidders.
Finally, this avenue of study is related to a literature on parametric auctions, in which a mechanism can depend on limited statistical information about the valuation distributions (such as the median, mean, and/or variance) \citep{AM12,ADMW13}.  The best-known parametric auctions for our setting are deterministic, but it is natural to ask whether randomization can help in this context.

%% file: prelim.tex
\paragraph{Single Item Auctions}

In a classical single item auction, the auctioneer has a single item to sell to
$n$~bidders.  Each bidder~$i$ has a private value~$\vali$ for receiving the
item.  We will consider the symmetric Bayesian setting, where each bidder's
value is drawn independently from a distribution~$\dist$ with density 
function~$\dens=\dist'$.  We write $\vals = (\vali[1], \dotsc, \vali[n])$ for the
profile of values.
%
%
An auction consists of \emph{allocation rules} $x_i: \mathbb R_+^n \to [0, 1]$ and 
\emph{payment rules} $p_i:  \mathbb R_+^n  \to \mathbb R$, meaning that, when the bid 
profile is $\vals$, each bidder~$i$ gets the item with probability
$x_i(\vals)$ and makes a payment of $p_i(\vals)$.  A bidder's utility is 
then $x_i(\vals) \vali - p_i(\vals)$.

The allocation rules have to satisfy the feasibility
constraint $\sum_i x_i(\vals) \leq 1$, for all $\vals$.  We also require an
auction to be \emph{individually rational}, that is, bidders always get
nonnegative returns.  Formally, for each bidder~$i$ and all bid profiles $\vals$,  $\vali x_i(\vals) - p_i(\vals) \geq 0$.  
An auction is \emph{Bayesian incentive compatible} if, for each bidder~$i$ and
value~$\vali$, the bidder's expected utility (over randomness in other bidders'
values)
is maximized by bidding her true valuation.  Formally, 
for each $i$, value $\vali$, and deviation $\vali'$, 
\begin{align*}
\Ex[\val_{-i}]{x_i(\vali, \vali[-i]) \vali - p_i (\vali, \vali[-i])} \geq 
\Ex[\val_{-i}]{x_i(\vali', \vali[-i]) \vali - p_i (\vali', \vali[-i])}.
\end{align*}

 
For each bidder~$i$, an allocation rule generates an \emph{interim allocation
rule} which maps her value~$\vali$ to a winning probability, in expectation over
the other bidders' bids.  We abuse notation and use $x_i(\vali) =
\Ex[\val_{-i}]{x_i(\vali, \vali[-i])}$ to denote the interim allocation rule.
The interim payment rule $p_i(\vali)$ is similarly defined.  Bayesian incentive
compability is therefore easily expressed by the interim allocation rules and
interim payment rules: for each bidder~$i$, each value~$\vali$ and possible
deviation~$\vali'$, we have $x_i(\vali) \vali - p_i(\vali) \geq x_i(\vali')
\vali - p_i(\vali')$.
\bjl{Should we just skip the prior definition of BIC, and use this one?}
\hf{I vote for keeping both, since this is a ``functional'' explanation of
interim allocations and payments, if the reader didn't get the ``structural''
explanation.}

The expected revenue of an auction is $\Ex[\vals]{\sum_i
p_i(\vals)}$.  Given a value distribution~$\dist$, a
\emph{revenue optimal} auction is one whose expected revenue is optimal among
all individually rational and Bayesian incentive compatible auctions.  We refer to the
revenue of this auction as the \emph{optimal revenue}. 

All auctions in this paper will satisfy the stronger \emph{(dominant strategy)
incentive compatibility} property, i.e., no matter what the the other bidders
bid, it is always in a bidder's best interest to truthfully bid her value.
Since we will only consider incentive compatible auctions (which is without loss
of generality by the revelation principle \citep{Mye81}), we have used the same
symbol~$\val$ for values and bids.  In general, an auction will be incentive
compatible if each bidder faces a take-it-or-leave-it price that does not depend
on her own bid.  All auctions considered in this paper will obviously satisfy
this condition.\footnote{We introduced the notion of Bayesian incentive
compatibility here because our benchmark optimal revenue mechanism needs only to
satisfy Bayesian incentive compatibility.  As Myerson showed, this is in fact
equal to the revenue of the optimal dominant strategy incentive compatible
auction.}

\paragraph{Revenue-Optimal Auctions}

In his seminal work, \citet{Mye81} laid the foundation for the study of 
revenue-optimal auctions.  The following theorem summarizes the part of his results that
will be used in this work.

\begin{theorem}[\citealp{Mye81}] In a single item auction where each bidder's
value is drawn i.i.d.\@ from a distribution~$\dist$, for any Bayesian incentive compatible auction with interim allocation rules $x_i$'s, 
\begin{enumerate}[(i)]
\item The expected revenue from each bidder is equal to the \emph{virtual
surplus}, defined as 
\begin{align*}
\Ex[\vali \sim \dist]{x_i (\vali) \cdot \left(\vali - \frac{1 - 
\dist(\vali)}{\dens(\vali)} \right)}.
\end{align*}
The term $\vali - \frac{1 - \dist(\vali)}{\dens(\vali}$ is called the
\emph{virtual value} of the value~$\vali$.  In other words, the virtual surplus
is the winning virtual value in expectation.

\item When the distribution~$\dist$ is \emph{regular}, i.e., when the virtual value is
monotone nondecreasing with the value, the optimal auction is the \emph{second
price auction with a reserve price~$\monres$}.  In this auction, the item is
only sold when at least one bid is above~$\monres$, and the winner pays the
higher of $\monres$ and the second highest bid.
\end{enumerate}
\end{theorem}

\paragraph{\citet{BR89}'s Interpretation and Revenue Curves}

In a classic work, \citet{BR89} gave an interpretation of \citet{Mye81}'s
optimal auciton and drew a connection between the theory of optimal auctions and the
theory of monopolist price discrimination.  This connection reveals much
economic intuition underlying \citeauthor{Mye81}'s results, and provides
powerful technical tools.  It is this viewpoint and tools that we will heavily
use in this work, and we explain this connection in some detail here.

A bidder whose value is drawn from a distribution~$\dist$ can be seen as a
market where the customers' values are distributed according to~$\dist$, which
gives rise to its demand curve.  In particular, if the monopolizer sets the
price of a good to sell at~$p$, then only customers whose value are above~$p$
will buy the good.  The demand is therefore $1 - \dist(p)$.  Each value~$\val$
is in this way mapped to its \emph{quantile} $\quant(\val) := 1 -
\dist(\val)$, its relative standing in this market.  The revenue collected when the monopolizer sells a
 quantity~$\quant$ is given by $\rev(\quant) := \val(\quant) \quant$, where $\val(\quant) := \dist^{-1}(1 - \quant)$.  This function $\rev: [0, 1] \to \mathbb R_{+}$
is called the \emph{revenue curve}.  Back from this analogy to the bidder, the
quantile~$\quant$ of a value~$\val$ is the probability with which the buyer will
buy at a take-it-or-leave-it price of~$\val$.  Note that for any
distribution~$\dist$, $\quant$ is uniformly distributed on $[0, 1]$.

\citet{BR89} showed that the slope of the revenue curve at a quantile~$\quant$
is exactly equal to the virtual value of the value corresponding to~$\quant$.
Formally, 
\begin{align*}
\rev'(\quant(\val)) = \frac{\dd \rev(\quant)}{\dd \quant} \Big|_{\quant =
\quant(\val)} = \val - \frac{1 - \dist(\val)}{\dens(\val)}.
\end{align*}
As a consequence, regular distributions, in which virtual values are monotone
nondecreasing with values, are exactly the distributions with concave
revenue curves.  Also, for an auction with interim allocation rule~$x_i$, the
virtual surplus (equal to the expected revenue) from a bidder is given by $\int_0^1
\rev'(\quant) x_i(\quant(\val)) \: \dd \quant$.

The highest point of a revenue curve corresponds to the optimal revenue one
could get by setting a take-it-or-leave-it price for a bidder.  This price we
call the \emph{monopoly reserve price} and denote by~$\monres$.  We denote the
quantile of~$\monres$ by~$\monresquant$.


Throughout the paper we will assume $\rev(0) = \rev(1) = 0$.  This is a rather
standard assumption in the literature and is without loss of generality; for
completeness we justify it in \autoref{sec:prelim-app}.

\paragraph{Analysis of Second Price Auctions by \citet{bu1996}}  \citet{bu1996}
showed that, for bidders whose values are drawn from a regular distribution, the
revenue of an $n+1$ bidder second price auction without reserve price is always
(weakly) better than the revenue of the $n$ bidder revenue optimal revenue auction.  This
immediately implies that for $n$ bidders, the revenue of the second price
auction without a reserve price is at least $1-\frac 1 n$ of the optimal
revenue.

%% file: span.tex
We will consider a variant of the second price auction, which will sometimes inflate the bid of the second-highest bidder before offering that bid as a fixed price for the highest bidder.

\begin{definition}
The \emph{$\blowfac$-inflated second price auction} offers the item to the
highest bidder at a take-it-or-leave-it price set as $(1 + \blowfac)$ times the
second highest bid.  The \emph{$(\blowprob, \blowfac)$-inflated second price
auction} runs the $\blowfac$-inflated second price auction with probability~$\blowprob$, and runs the second price auction with probability $1 - \blowprob$.
\end{definition}

We will prove the following main theorem in this section.

\begin{theorem}
\label{thm:span}
For any $n$, there is a $(\blowprob, \blowfac)$-inflated second price auction,
such that for any $n$ bidders with values drawn i.i.d.\@ from a regular
distribution, the revenue of the inflated second price auction is strictly larger than $\tfrac
{n-1}n$ fraction of the optimal revenue.
\end{theorem}

At the end of the section, we give the improved approximation ratios of $0.512$
for the case $n = 2$.

We will begin our analysis by studying the relationship between the optimal
revenue and the revenue of the second price auction (SPA) without reserve.
The goal of this analysis is to refine the standard $\tfrac n
{n-1}$-approximation result and to present an approximation that depends on the
quantile of the optimal reserve price, $\monresquant$.

We begin by deriving a bound on the additive revenue loss suffered by
using SPA rather than the optimal auction.

\begin{lemma}
\label{lem:spanloss}
For $n$ bidders with values drawn i.i.d.\@ from a regular distribution, the
difference between the revenue of the optimal auction and that of the second
price auction is at most $\rev(\monresquant) (1 - \monresquant)^{n-1}$.
\end{lemma}
\ps{do we need to repeat the assumptions here, i.e. could we omit $F$ regular $v_i\sim F$}
\hf{I vote for keeping the assumptions here.}
\begin{proof}
The optimal auction is a second price auction with a reserve at $\monres$.  The
allocation rules for the optimal auction and the second price auction
therefore differ only when the highest value is below $\monres$, i.e., its
corresponding quantile is larger than $\monresquant$.  Such quantiles correspond
to negative virtual values, and make up the difference between the
optimal auction (where the contribution is zero) and the second price auction (where the
contribution is negative).  For a bidder bidding at quantile~$\quant$, 
the probability that it is the lowest quantile is $(1 - \quant)^{n -
1}$.  The total negative virtual surplus generated by one bidder, 
over quantiles above $\monresquant$, is therefore
\begin{align*}
& \int_{\monresquant}^1 \rev'(\quant) (1 - \quant)^{n-1} \: \dd \quant \\
= & \rev(\quant) (1 - \quant)^{n-1} \Big|_{\monresquant}^1 + (n-1)
\int_{\monresquant}^1 \rev(\quant) (1 - \quant)^{n-2} \: \dd \quant \\
= & -\rev(\monresquant)(1 - \monresquant)^{n-1} + (n-1)
\int_{\monresquant}^1 \rev(\quant) (1 - \quant)^{n-2} \: \dd \quant
\end{align*}

This is the negative of the difference between the optimal revenue and the
revenue of the second price auction.  To get an upper
bound of the difference, we need to find a lower bound for the integral.  Using the
concavity of $\rev(\quant)$, we know that, for $\quant \geq \monresquant$,
$\rev(\quant) \geq \rev(\monresquant) \cdot \frac{1- \quant}{1 - \monresquant}$.
Therefore, we know that the quantity above is at most
\begin{align*}
-\rev(\monresquant)(1 - \monresquant)^{n-1} + (n-1) \int_{\monresquant}^1
\rev(\monresquant) \frac{1 - \quant}{1 - \monresquant} \cdot (1 - \quant)^{n-2}
\: \dd \quant = - \frac 1 n \rev(\monresquant) (1 - \monresquant)^{n-1}.
\end{align*}

This is the revenue difference due to each bidder.  Multiplying this by $n$ gives us the
lemma.
\end{proof}

As a corollary, we obtain the following bound on the ratio between the revenue
of SPA and the optimal revenue.

\begin{corollary}
\label{cor:spanloss}
For $n$ i.i.d.\@ bidders with a regular valuation distribution, for any
$\monresquant \in (0, 1]$, the second price auction extracts at least a
$\frac{1 - (1 - \monresquant)^{n-1}}{1 - (1 - \monresquant)^n}$
fraction of the optimal revenue.  
\end{corollary}

\begin{proof}
We give a lower bound on the optimal revenue, which, combined with
Lemma~\ref{lem:spanloss}, gives a lower bound on the approximation ratio of the
SPA.  The optimal auction could post $\monres$ as a take-it-or-leave-it price to
each bidder in turn, and sell at that price to the first bidder who accepts.
This gives a revenue of $\rev(\monresquant)[1 + (1 - \monresquant) + (1 -
\monresquant)^2 + \cdots + (1 - \monresquant)^{n-1}]$.  The ratio of the revenue
of the second price auction to the optimal revenue is therefore at least
\begin{align*}
1 - \frac{(1 - \monresquant)^{n-1}}{1 + (1 - \monresquant) + \cdots + (1 -
\monresquant)^{n-1}} = \frac{1 + (1 - \monresquant) + \cdots + (1 -
\monresquant)^{n-2}} {1 + (1 - \monresquant) + \cdots + (1 -
\monresquant)^{n-1}} = \frac{1 - (1 - \monresquant)^{n-1}}{1 - (1 -
\monresquant)^n}.
\end{align*}
\end{proof}

Note that the factor in Corollary \ref{cor:spanloss} is a strictly increasing
function in $\monresquant$, and is equal to $\frac{n-1}{n}$ when $\monresquant =
0$.  For the case $n = 2$, we see that the approximation ratio of the SPA is at
least $\frac{1}{2 - \monresquant}$.

We will now analyze the
$\blowfac$-inflated second price auction, and show that its approximation ratio is
better when $\monresquant$ is small.  Before that, we first prove a technical
lemma that gives us bounds on quantiles in a regular distribution for values
that are apart by a multiplicative factor.


\begin{lemma}
\label{lem:quant-bound}
For any regular distribution:
\begin{enumerate}
\item $\quant\left(\frac{\monres}{1 + \blowfac}\right) \leq
(1 + \blowfac) \monresquant$, and
\item 
For any $\quantquestion \geq \monresquant$,
$\quant\left(\frac{\val(\quantquestion)}{1 + \blowfac}\right) \geq \frac{1 +
\blowfac}{1 + \blowfac \quantquestion} \quantquestion$.  \end{enumerate}
\end{lemma}

\begin{proof}
The first statement is a direct consequence of the optimality of $\monres$ as a
reserve price for a single bidder: $ \monres
\monresquant = \rev(\monresquant) \geq \rev(\quant(\monres / (1 + \blowfac))) =
\quant(\monres / (1 + \blowfac)) \cdot \monres / (1 + \blowfac)$.

For the second statement, let us denote by $\quant_1$ the quantile of
$\val(\quantquestion) / (1 + \blowfac)$.  By regularity of the distribution, the revenue
curve is concave, and because $\quant_1$ is greater than $\quantquestion$ and
both are greater than $\monresquant$, on the revenue curve the point $(\quant_1,
\rev(\quant_1))$ is above the straight
line connecting $(1, \rev(1) = 0)$ and $(\quantquestion, \rev(\quantquestion))$.
Therefore,
\begin{align*}
\frac{\rev(\quantquestion)}{1 - \quantquestion} \leq \frac{\rev(\quant_1)}{1 -
\quant_1}, \quad \Rightarrow \quad
\frac{\val(\quantquestion)\quantquestion}{1 - \quantquestion} \leq
\frac{\quant_1 \val(\quantquestion)}{(1 + \blowfac)(1 - \quant_1)}.
\end{align*}
Rearranging the terms we have that $\quant_1 \geq \frac{1 + \blowfac}{1 +
\blowfac \quantquestion} \quantquestion$.  
\end{proof}

The next lemma lower bounds the approximation ratio of $\blowfac$-inflated SPA for distributions with a small $\monresquant$.

\begin{lemma}
\label{lem:delta-blown-n}
For a regular distribution with $\monresquant < 1/n$, the revenue of 
the $\blowfac$-inflated SPA for $n$ i.i.d.\@ bidders is at least 
\begin{align*}
n \rev(\monresquant) \left[ [1 - (1 + \blowfac) \monresquant]^{n-1} - \left(1 -
\frac{(1 + \blowfac) \monresquant}{1 + \blowfac \monresquant} \right)^{n-1} +
\frac{(n-1)(1 + \blowfac)}{1 - \monresquant} \int_{\monresquant}^1 \frac{(1 -
\quant)^{n-1}}{(1 + \blowfac \quant)^n} \: \dd \quant \right].
\end{align*}
\end{lemma}

\begin{proof}
Let the interim allocation rule of the $\blowfac$-inflated second price auction
for a bidder~$i$ be $x_i$.  Recall that the revenue from a bidder~$i$ is given
by $\int_0^1 \rev'(\quant) x_i(\quant) \: \dd \quant$.  On $[0, \monresquant]$,
$\rev'(\quant)$ is positive, and we need to lower bound $x_i(\quant)$; on
$(\monresquant, 1]$, $\rev'(\quant)$ is negative, and we need to upper bound
$x_i(\quant)$.

We first lower bound $\int_0^{\monresquant} \rev'(\quant) x_i(\quant) \: \dd
\quant$.  For such quantiles, the corresponding value is larger than $\monres$.
Obviously, if any bidder has a
value~$\val > \monres$, as long as all other
bidders' values are below $\frac{\monres}{1 + \blowfac}$, the bidder
bidding~$\val$ will win.  By Lemma~\ref{lem:quant-bound}, the quantile of $\frac
{\monres}{1 + \blowfac}$ is at most $(1 + \blowfac) \monresquant$, and
therefore the probability that a bidder's value is below $\frac{\monres}{1 +
\blowfac}$ is at least $1 - (1 + \blowfac) \monresquant$.  Hence, the expected
revenue collected from each bidder for values larger than the monopoly
reserve is at least
\begin{align*}
\int_0^{\monresquant} \rev'(\quant) [1 - (1 + \blowfac) \monresquant]^{n-1} \:
\dd \quant = [1 - (1 + \blowfac)\monresquant]^{n-1} \rev(\monresquant).
\end{align*}

We then upper bound the negative of $\int_{\monresquant}^1 \rev'(\quant)
x_i(\quant) \: \dd \quant$.  For such quantiles, the corresponding value is smaller than~$\monres$.
Recall that we would like to upper bound $x_i(\quant)$.
By the definition of $\blowfac$-inflated
SPA, such a value wins the auction if and only if all other bidders bid below
$\val / (1 + \blowfac)$.  By Lemma~\ref{lem:quant-bound}, the quantile of $\val /
(1 + \blowfac)$ is at least $\frac{1 + \blowfac}{1 + \blowfac \quant} \quant$.
In other words, the probability that an independent draw has value less than
$\frac{\val}{1 + \blowfac}$ is at most $1 - \frac{1 + \blowfac}{1 + \blowfac
\quant} \quant = \frac{1 - \quant}{1 + \blowfac \quant}$.  Therefore, the
negative contribution to the virtual surplus by each bidder is lower bounded by
\begin{align*}
\int_{\monresquant}^1 \rev'(\quant) \left( \frac{1- \quant}{1 + \blowfac \quant}
\right)^{n-1} \: \dd \quant 
= & \rev(\quant) \left( \frac{1- \quant}{1 + \blowfac \quant}
\right)^{n-1}
\Big|_{\monresquant}^1 + (n-1) \int_{\monresquant}^1 \rev(\quant) \left( \frac{1
- \quant}{1 + \blowfac \quant} \right)^{n-2} \cdot \frac{1 + \blowfac}{(1 + \blowfac
\quant)^2} \: \dd \quant \\
= & - \rev(\monresquant) \left(\frac{1 - \monresquant}{1 + \blowfac
\monresquant} \right)^{n-1} + (n-1)(1 + \blowfac) \int_{\monresquant}^1
\rev(\quant) \frac {(1 - \quant)^{n-2}}{(1 + \blowfac \quant)^n} \: \dd \quant,
\end{align*}
where in the first step we did an integration by part, and in the second step we
used the fact $\rev(1) = 0$.  Since we aim to lower bound this quantity, we
use the fact that $\rev(\quant)$ to the right of~$\monresquant$ is pointwise lower bounded by
$\rev(\monresquant) \cdot \frac{1 - \quant}{1 - \monresquant}$ because of its
concavity.  Substituting this, we have that
\begin{align}
& \int_{\monresquant}^1 \rev(\quant) \frac {(1 - \quant)^{n-2}}{(1 + \blowfac \quant)^n} \: \dd \quant 
\geq \int_{\monresquant}^1 \rev(\monresquant) \cdot \frac{1 - \quant}{1 -
\monresquant} \cdot \frac{(1 - \quant)^{n-2}}{(1 + \blowfac \quant)^n} \: \dd
\quant = \frac{\rev(\monresquant)}{1 - \monresquant} \int_{\monresquant}^1 \frac{(1 - \quant)^{n-1}}{(1 +
\blowfac \quant)^n} \: \dd \quant.
\label{eq:thatintegral}
\end{align}

Combining everything together, the revenue of the $\blowfac$-inflated second price
auction is at least
\begin{align}
n \rev(\monresquant) \left[ [1 - (1 + \blowfac) \monresquant]^{n-1} - \left(1 -
\frac{(1 + \blowfac) \monresquant}{1 + \blowfac \monresquant} \right)^{n-1} +
\frac{(n-1)(1 + \blowfac)}{1 - \monresquant} \int_{\monresquant}^1 \frac{(1 -
\quant)^{n-1}}{(1 + \blowfac \quant)^n} \: \dd \quant \right].
\label{eq:blown-apx}
\end{align}
\end{proof}

\begin{lemma}
\label{lem:delta=1}
For $\blowfac = 1$ and $\monresquant < \tfrac 1 n$, the ratio of the
$1$-inflated second price auction revenue to the optimal revenue is at least
\begin{align*}
& [1 - (1 + \blowfac) \monresquant]^{n-1} - \left(1 -
\frac{(1 + \blowfac) \monresquant}{1 + \blowfac \monresquant} \right)^{n-1}
+ \frac {(n-1)(1 - \monresquant)^{n-1}}{n(1 + \monresquant)^n} \\
& + \frac {n-1}
{n(1-\monresquant)} \left[ \left(1 - \frac 1 {n^2} \right)^{-n} - \frac{1}{(1 - \monresquant^2)^n} \right] \left(1 - \frac 1 n \right)^{2n}. 
\end{align*}
\end{lemma}
We relegate the proof to \autoref{sec:span-app}.

\begin{proof}[Proof of \autoref{thm:span}]
By Lemma~\ref{lem:delta=1}, for $\monresquant = 0$, the approximation ratio of the
$1$-inflated second price auction is at least $\frac{n-1}{n} \cdot [1 + (1 - \frac{2}{n+1})^n]$, which
is strictly greater than $\frac {n-1}{n}$ for any~$n$.  Since the bound given by
Lemma~\ref{lem:delta=1} is a continuous function, for a sufficiently small
$\underline \quant > 0$, for all $\monresquant < \underline \quant$, the
$1$-inflated SPA has an approximation ratio at least $\frac{n-1}{n} \cdot [1 +
\frac 1 2 (1 - \frac{2}{n+1})^n]$.  Recall that the approximation ratio we
derived in \autoref{cor:spanloss} is a strictly increasing function
in~$\monresquant$ that equals to $\frac {n-1}{n}$ at $\monresquant = 0$.
Therefore at $\underline \quant$, the approximation ratio of the SPA is
$\frac{n-1}{n} (1 + \gamma)$ for some $\gamma > 0$.  Taking $\blowprob > 0$ small
enough such that $(1 + \gamma) (1 - \blowprob) > 1 + \eta$ for some $\eta > 0$,
we will ensure that for all $\monresquant > \underline \quant$, the
approximation of the $(\blowprob, \blowfac)$-inflated SPA is at least $\frac
{n-1}{n} (1 + \eta)$ (this pessimistically does not assume any revenue coming
from the $1$-inflated SPA).  On the other hand, for $\monresquant \leq
\underline \quant$, since the approximation ratio of the SPA is always at least
$\frac {n-1}{n}$, the $(\blowprob, \blowfac)$-inflated SPA has an approximation
ratio at least $\frac{n-1}{n} (1 + \frac{\blowprob}{2} (1 - \frac 2 {n+1})^n)$.
This proves the theorem.
\end{proof}

\begin{remark}
\label{rem:spatwo}
Lemma~\ref{lem:delta=1} is in place because the integral
in~\eqref{eq:thatintegral} is not easy for general values of~$n$.  The argument
in the proof of \autoref{thm:span} is rather pessimistic.  
However, for concrete values of~$n$, one can compute the revenue lower
bound in Lemma~\ref{lem:delta-blown-n} for any $\blowfac > 0$, without going
through further losses in the analysis of Lemma~\ref{lem:delta=1} and \autoref{thm:span}, and get better approximation ratios.  For
example, for $n = 2$ and $\blowfac = 1$,
\begin{align*}
\int_{\monresquant}^1 \frac{(1 - \quant)^{n-1}}{(1 + \blowfac \quant)^n} \: \dd
\quant = \int_{\monresquant}^1 \frac{1-\quant}{(1 + \quant)^2} \: \dd \quant = 
\frac{2}{\monresquant+1}+\log \left(\frac{\monresquant+1}{2}\right)-1.
\end{align*}
One can substitute this into \eqref{eq:blown-apx} and combine it with
\autoref{cor:spanloss}; numerical computation shows that the $(0.15,
1)$-inflated second price auction gives at least $0.512$ fraction of the optimal
revenue.
\end{remark}

%% file: singlesample.tex
In this section we show that, for any buyer with value drawn from a regular
distribution~$\dist$, with one sample from the same distribution, one can extract
strictly more than half of the the optimal revenue, by introducing
randomization in the use of the sample.  We denote the sample by $\sample$, and the buyer's value
by~$\val$.  Note that in this setting, the optimal revenue is simply
$\rev(\monresquant)$.

\begin{definition}
The post-the-sample algorithm posts the sample~$\sample$ as a take-it-or-leave-it
price.  The $\shadefac$-shaded post-the-sample algorithm posts $(1 -
\shadefac)\sample$ as a take-it-or-leave-it price.  The $\blowfac$-inflated
post-the-sample algorithm posts $(1 + \blowfac)\sample$ as the take-it-or-leave-it
price.  The $(\shadeprob, \shadefac, \blowprob, \blowfac)$-randomized
post-the-sample algorithm runs $\shadefac$-shaded post-the-sample with
probaiblity~$\shadeprob$, $\blowfac$-inflated post-the-sample with
probability~$\blowprob$, and (normal) post-the-sample with probability $1 -
\shadeprob - \blowprob$.
\end{definition}

\begin{theorem}
\label{thm:singlesample}
There exists $\shadeprob, \shadefac, \blowprob, \blowfac$ such that for any
regular distribution, with a single sample, the $(\shadeprob, \shadefac, \blowprob,
\blowfac)$-randomized post-the-sample algorithm extracts strictly more than half
of the optimal revenue.
\end{theorem}

Using our randomized post-the-sample algorithm in place of the original
post-the-sample algorithm in the auction of \citet{dhangwatnotai2014revenue}, we
have the following corollary.

\begin{corollary}
\label{cor:dhangwatnotai}
In a single-item multi-bidder auction, where bidders' values are drawn
independently from regular distributions, and where for each bidder there is at
least another bidder whose value is drawn from the same distribution, there is a
prior-independent auction whose revenue is strictly better than $\tfrac 1 8$ of
the revenue of the optimal auction which knows all distributions.
\end{corollary}

In \autoref{thm:singlesample}, the improvement over the one half approximation will be on the order of $10^{-8}$, and this is admittedly mainly
of theoretical interest (we also made no effort in optimizing the parameters).  The concrete values of $\shadeprob$ and $\blowprob$
are given in the proof near the end of the section.  We first analyze the
performance of the three ingredient mechanism in the randomized post-the-sample
algorithm, given in the next three lemmas.

\begin{lemma}
For any $\monresquant \in [0, \tfrac 1 2]$, the approximation ratio of the
$1$-inflated post-the-sample algorithm is at least 
\begin{align}
- 2\monresquant + \frac{2\monresquant}{1 + \monresquant} + \frac{2}{1 - \monresquant} \left[ \frac{2}{\monresquant+1}+\log \left(\frac{\monresquant+1}{2}\right)-1 \right].
\label{eq:blowsample}
\end{align}
\end{lemma}

\begin{proof}
The revenue from a single bidder by posting twice the sample is exactly half of
the revenue from a two bidder $1$-inflated second price auction where the two
bidders values are drawn i.i.d.\@ from the same regular distribution.  The lemma
then follows from Lemma~\ref{lem:delta-blown-n} and \autoref{rem:spatwo}.
\end{proof}

\begin{lemma}
\label{lem:shade}
For any regular distribution, if posting a price of $\frac{\monres}{1 -
\shadefac}$ obtains a revenue that is $\revshade$ fraction of the optimal
revenue $\rev(\monresquant)$, then the revenue of the $\shadefac$-shaded post-the-sample algorithm is at least
\begin{align}
\left[(1 - \shadefac)\left( \frac {\monresquant + 1} 2 - \revshade \monresquant \right) + \monresquant \revshade \shadefac \right] \rev(\monresquant).
\label{eq:shadesample}
\end{align}
\end{lemma}
This is the most technically involved proof of this section.  We relegate it to
\autoref{sec:singlesample-app}.

\begin{lemma}
\label{lem:sample-revshade}
If for a regular distribution, posting a price of $\frac{\monres}{1 -
\shadefac}$ obtains a revenue that is $\revshade$ fraction of the optimal
revenue $\rev(\monresquant)$, then the revenue of the post-the-sample algorithm
is at least $\frac 1 2 (1 + \monresquant \shadefac \revshade) \rev(\monresquant)$.
\end{lemma}

\begin{proof}
The revenue of the post-the-sample algorithm is equal to the total area under
the revenue curve \citep{dhangwatnotai2014revenue}.  We will therefore give a lower bound on this
area.  In general, subject to concavity and fixing~$\rev(\monresquant)$, the
area under the revenue curve is minimized at $\tfrac 1 2 \rev(\monresquant)$
by the triangle.  Given $\revshade$, the fraction of~$\rev(\monresquant)$
extracted by posting a price $\monres/(1 - \shadefac)$, the quantile of
$\monres / (1 - \shadefac)$ is given by $\revshade \rev(\monresquant) (1 -
\shadefac)$.  At this quantile, the triangle revenue curve would have a revenue
of $\rev(\monresquant) \revshade (1 - \shadefac)$, but the current revenue is
$\rev(\monresquant) \revshade$.  The two differ by $\shadefac \revshade
\rev(\monresquant)$.  The extra area over the triangle is therefore at least
$\tfrac 1 2 \monresquant \shadefac \revshade \rev(\monresquant)$.
\end{proof}

\begin{proof}[Proof of \autoref{thm:singlesample}]
We will show that setting $\blowfac = 1$, $\shadefac = 0.01$, and sufficiently
small $\blowprob$ and~$\shadeprob$, with $\blowprob = 4\shadeprob$, will
guarantee a revenue better than $\tfrac 1 2$ of the optimal revenue.  Note that
\eqref{eq:shadesample} is a decreasing function in~$\revshade$ for $\shadefac <
0.5$.  For $\monresquant \leq 0.02$ and any value of~$\revshade$, $0.8
\times$\eqref{eq:shadesample}$+ 0.2 \times$\eqref{eq:blowsample} is at least
$0.505 \rev(\monresquant)$, by taking the worst value of $\revshade = 1$.  For $\revshade \leq 0.05$ and any value
of~$\monresquant$, $0.8 \times$\eqref{eq:shadesample}$+ 0.2
\times$\eqref{eq:blowsample} is at least $0.518 \rev(\monresquant)$, by taking
the worst value of $\revshade = 0.05$, and the minimum is taken at
$\monresquant = 0$.  For the only remaining case, i.e., $\monresquant > 0.02$
\emph{and} $\revshade > 0.05$,  Lemma~\ref{lem:sample-revshade} gives that the
revenue of the post-the-sample algorithm is at least $0.500005
\rev(\monresquant)$.  Therefore, running post-the-sample with probability $1 -
10^{-6}$ guarantees a revenue of better than $0.5 + 10^{-6}$ fraction of
$\rev(\monresquant)$ in this case.  With the remaining probability of~$10^{-6}$,
running the $\blowfac$-inflated post-the-sample with probability $2 \times
10^{-7}$ and the $\shadefac$-shaded post-the-sample with probability $8 \times
10^{-7}$ guarantees a revenue of $0.5 + 5 \times 10^{-9}$ fraction of the
optimal revenue in the cases analyzed above.  Overall, the $(0.8 \times 10^{-6},
0.01, 0.2 \times 10^{-6}, 1)$-randomized post-the-sample algorithm gives an
approximation ratio of at least $0.5 + 5 \times 10^{-9}$.

\end{proof}

%% file: prelimapp.tex
Throughout the paper we have made the assumption $\rev(0) = \rev(1) = 0$.  This is standard in the literature, and we briefly show that this is without loss of generality.

A value distribution whose support includes $0$ satisfies $\rev(1) = 0$; otherwise, one can mix into the distribution with probability~$\eps$ a uniform distribution on $[0, \underline{\val}]$, where $\underline{\val}$ is the infimum of the support.  As $\eps$ approaches $0$, this mixture converges to the original distribution, and for sufficiently small~$\eps$, the resulting revenue curve remains concave if the original revenue curve is.  As for the assumption ~$\rev(0) = 0$, we first justify it for bounded support distributions.  When there is no point mass on the supremum of the support, its revenue will be~$0$ at quantile~$0$.  If this is not the case, we can mix into the distribution with probability~$\eps$ a uniform distribution on $[\bar \val, \bar \val + \delta]$, where $\bar \val$ is the supremum of the support, and $\delta > 0$ is a small positive real number.  For sufficiently small $\delta$ and $\eps$, the resulting revenue curve will be concave if the original one is; as $\eps$ approaches~$0$, the mixture also converges to the original distribution.  Any unbounded distribution's revenue can be approached arbitrarily well by taking its truncation at higher and higer values.  The truncated distribution is bounded.  All our analysis would not be affected by all such asymptotic approximations, and we will simply assume $\rev(0) = \rev(1) = 0$.

%% file: spanapp.tex
\begin{proof}[Proof of Lemma~\ref{lem:delta=1}]
We estimate the integral in \eqref{eq:thatintegral}.  For $\blowfac = 1$,
\begin{align*}
\int_{\monresquant}^1 \frac{(1 - \quant)^{n-1}}{(1 + \blowfac \quant)^n} \: \dd \quant 
= & \int_{\monresquant}^1 \frac{(1 - \quant)^{n-1}}{(1 + \quant)^n} \: \dd \quant = \int_{\monresquant}^1 \frac{(1 - \quant)^{2n-1}}{(1 - \quant^2)^n} \: \dd \quant \\
= & \int_{\monresquant}^{\frac 1 n} \frac{(1 - \quant)^{2n-1}}{(1 - \quant^2)^n}
\: \dd \quant + \int_{\frac 1 n}^1 \frac{(1 - \quant)^{2n-1}}{(1 - \quant^2)^n}
\: \dd \quant \\
\geq & \frac{1}{(1 - \monresquant)^n} \int_{\monresquant}^{\frac 1 n}  (1 - \quant)^{2n-1} \: \dd \quant +
\int_{\frac 1 n}^1 \frac{(1 - \quant)^{2n-1}}{\left(1 - \frac 1 {n^2}
\right)^{n}} \: \dd \quant  \\
= & \frac {(1 - \monresquant)^{n}}{2n(1 + \monresquant)^n} + \frac 1 {2n} \left[
\left(1 - \frac 1 {n^2} \right)^{-n} - \frac{1}{(1 - \monresquant^2)^n} \right] \left(1 - \frac 1 n \right)^{2n}. 
 \end{align*}
Substituting this into \eqref{eq:blown-apx}, and noting that the optimal
revenue is at most~$n \rev(\monresquant)$, we obtain the lemma.
\end{proof}

%% file: singlesampleapp.tex
\begin{proof}[Proof of Lemma~\ref{lem:shade}]
We first introduce
some notations.  Let $\quantshade: [0, 1] \to [0, 1]$ be a function that maps a
quantile~$\quantquestion$ to the quantile whose value is $1 / (1 - \shadefac)$ times the
value corresponding to~$\quantquestion$.  (Formally,
$\quantshade(\quantquestion) = \quant(\dist^{-1}(1 - \quantquestion) / (1 - \shadefac))$.\footnote{If $\val / (1 -
\shadefac)$ is beyond the support of the distribution, just let $\quantshade$
be~$0$.}) 

We consider the difference of $\shadefac$-shaded post-the-sample
algorithm as compared with the plain post-the-sample algorithm. 
The two algorithms differ only when $\sample$ falls in the interval $[\val,
\frac{\val}{1 - \shadefac}]$: the post-the-sample algorithm does not serve,
whereas the $\shadefac$-shaded post-the-sample algorithm does.  Conditioning
on~$\val$, this event happens with probability $\quant(\val) -
\quantshade(\quant(\val))$.  The revenue of the post-the-sample algorithm is
$\int_0^1 \rev'(\quant) (1 - \quant) \: \dd \quant$, and therefore the revenue
of the $\shadefac$-shaded post-the-sample algorithm is
\begin{align*}
\int_0^1 \rev'(\quant) (1 - \quant) \: \dd \quant + \int_0^1 \rev'(\quant)
(\quant - \quantshade(\quant) \: \dd \quant = \int_0^1 \rev'(\quant) \: \dd
\quant - \int_0^1 \rev'(\quant) \quantshade(\quant) \: \dd \quant = - \int_0^1
\rev'(\quant) \quantshade(\quant) \: \dd \quant.
\end{align*}
We need to lower bound this term.

When the buyer's value~$\val$ is smaller than $\monres$, its virtual value is
negative.  Here we would like to lower bound $\quantshade(\quant)$.  Let $\revshadedown$ be the ratio of the revenue of posting $\monres (1 - \shadefac)$ to the optimal revenue $\rev(\monresquant)$.  Let $\quantshadedown$ be the quantile of the value $\monres(1 - \shadefac)$, then because $\monres(1 - \shadefac) \quantshadedown = \revshadedown \monres \monresquant$, we have $\quantshadedown = \revshadedown \monresquant / (1 - \shadefac)$.
For $\val \leq \monres(1 - \shadefac)$, we know that the revenue of posting the
price~$\val$ is smaller than that of posting $\val / (1 - \shadefac)$.
In other words, $\val \quant(\val) \leq \frac{\val}{1 - \shadefac}
\quantshade(\quant(\val))$, and therefore we get the lower bound
$\quantshade(\quant) \geq (1 - \shadefac) \quant$, $\forall \quant \in
[\quantshadedown, 1]$.

Let the revenue of posting $\frac \monres {1 - \shadefac}$ be $\revshade
\rev(\monresquant)$.  Then the quantile of the value $\frac \monres {1 -
\shadefac}$ is $\revshade \monresquant (1 - \shadefac)$.  We know that,
for all $\quant \in [\monresquant, \quantshadedown]$, $\quantshade \geq
\revshade \monresquant (1 - \shadefac)$.  Now we have that
\begin{align*}
& - \int_0^1 \rev'(\quant) \quantshade(\quant) \: \dd \quant \geq
-\int_0^{\monresquant} \rev'(\quant) \quantshade(\quant) \: \dd \quant -
\int_{\monresquant}^{\quantshadedown} \rev'(\quant) \revshade \monresquant(1 -
\shadefac) \: \dd \quant - (1 - \shadefac) \int_{\quantshadedown}^1
\rev'(\quant) \quant \: \dd \quant \\
= & -\int_0^{\monresquant} \rev'(\quant) \quantshade(\quant) \: \dd \quant -
\revshade \monresquant(1 - \shadefac) \rev(\quant) \Big|_{\monresquant}^{\quantshadedown} - (1 - \shadefac) \left[ \rev(\quant) \quant
\Big|_{\quantshadedown}^1 - \int_{\quantshadedown}^1 \rev(\quant) \: \dd \quant
\right] \\
\end{align*}
Using $\rev(\quantshadedown) = \revshadedown \rev(\monresquant)$, this is equal
to
\begin{align*}
& -\int_0^{\monresquant} \rev'(\quant) \quantshade(\quant) \: \dd \quant - 
(1 - \shadefac) \left[\revshade \monresquant (\revshadedown - 1)
\rev(\monresquant) - \revshadedown \rev(\monresquant) \cdot \frac {\revshadedown
\monresquant}{1 - \shadefac} - \int_{\quantshadedown}^1 \rev(\quant) \: \dd
\quant \right] \\
= & -\int_0^{\monresquant} \rev'(\quant) \quantshade(\quant) \: \dd \quant  
+ (1 - \shadefac) \revshade \monresquant (1 - \revshadedown)
\rev(\monresquant) + \revshadedown^2 \monresquant \rev(\monresquant) + (1 -
\shadefac) \int_{\quantshadedown}^1 \rev(\quant) \: \dd \quant \\
\geq & -\int_0^{\monresquant} \rev'(\quant) \quantshade(\quant) \: \dd \quant  
+ (1 - \shadefac) \revshade \monresquant (1 - \revshadedown)
\rev(\monresquant) + \revshadedown^2 \monresquant \rev(\monresquant) + (1 -
\shadefac) \revshadedown \rev(\monresquant) \cdot \frac 1 2 \left(1 -
\frac{\revshadedown \monresquant}{1 - \shadefac}
\right) \\
= & -\int_0^{\monresquant} \rev'(\quant) \quantshade(\quant) \: \dd \quant  
+ (1 - \shadefac) \revshade \monresquant (1 - \revshadedown)
\rev(\monresquant) + \revshadedown^2 \monresquant \rev(\monresquant) +
\frac {\revshadedown \rev(\monresquant)} 2 (1 - \shadefac - \revshadedown
\monresquant),
\end{align*} 
where in the inequality we used the concavity of $\rev$ and the fact
$\rev(\quantshadedown) = \revshadedown \rev(\monresquant)$.  Note that the first
term is a quantity unaffected by the value of~$\revshadedown$.  Therefore we can
take the partial derivative of this lower bound with respect to $\revshadedown$
and get
\begin{align*}
\rev(\monresquant) \left[ - (1 - \shadefac) \revshade \monresquant +
\monresquant \revshadedown + \frac{1 - \shadefac} 2  \right].
\end{align*}

Since $\quantshadedown \geq \monresquant$, $\monres(1 - \shadefac)
\quantshadedown \geq (1 - \shadefac) \monres \monresquant$, and hence
$\revshadedown \geq 1 - \shadefac$.\footnote{A tighter lower bound for
$\revshadedown$ would be $\frac{1 - \shadefac}{1 - \shadefac + \shadefac
\monresquant}$, but the loose bound $1 - \shadefac$ will suffice for our
purpose.}  Therefore the partial derivative is lower bounded by
\begin{align*}
\rev(\monresquant) (1 - \shadefac) \left(\monresquant (1 - \revshade) + \frac
1 2 \right) \geq 0.
\end{align*}

Therefore, to minimize the revenue, the adversary should minimize
$\revshadedown$.  Hence, substituting $\revshadedown$ by $1 - \shadefac$, we get
the following lower bound on the revenue:
\begin{align}
-\int_0^{\monresquant} \rev'(\quant) \quantshade(\quant) \: \dd \quant +
\left[\frac {(1 - \shadefac)(\monresquant + 1)} 2 + \monresquant \revshade
\shadefac \right] \rev(\monresquant).
\label{eq:shade-bound-right}
\end{align}

Now we consider bounding the first term.  For $\quant \leq \monresquant$, its
virtual value $\rev'(\quant)$ is positive, therefore we need to upper bound
$\quantshade(\quant)$.

Recall that $\revshade$ is the number such that the revenue of posting $\frac
{\monres}{1 - \shadefac}$ is $\revshade \rev(\monresquant)$.  Let $\quant_1$
denote the quantile of the value $\monres / (1 - \shadefac)$, and we knew that
$\quant_1 = \revshade \monresquant (1 - \shadefac)$.
Notice that, for every $\quant \in [0, \monresquant]$, $\quantshade(\quant) \geq \quant_1$.  
We have 
\begin{align*} 
\int_0^{\monresquant} \rev'(\quant) \quantshade(\quant) \: \dd \quant \leq
\quant_1 \int_0^{\monresquant} \rev'(\quant) \: \dd \quant = \revshade
\monresquant (1 - \shadefac) \rev(\monresquant).
\end{align*} 

Substituting this to \eqref{eq:shade-bound-right}, we see that the revenue is
lower bounded by
\begin{align*}
\left[(1 - \shadefac)\left( \frac {\monresquant + 1} 2 - \revshade \monresquant \right) + \monresquant \revshade \shadefac \right] \rev(\monresquant).
\end{align*}
\end{proof}

%% file: main.bbl
\begin{thebibliography}{}

\bibitem[Azar et~al., 2013]{ADMW13}
Azar, P., Daskalakis, C., Micali, S., and Weinberg, S.~M. (2013).
\newblock Optimal and efficient parametric auctions.
\newblock In {\em Proceedings of the Twenty-Fourth Annual ACM-SIAM Symposium on
  Discrete Algorithms}, pages 596--604. SIAM.

\bibitem[Azar and Micali, 2012]{AM12}
Azar, P. and Micali, S. (2012).
\newblock Optimal parametric auctions.
\newblock Mit-csail-tr-2012-011.

\bibitem[Bergemann and Schlag, 2011]{bergemann2011should}
Bergemann, D. and Schlag, K. (2011).
\newblock Should i stay or should i go? search without priors.
\newblock mimeo, available at:
  http://www.gtcenter.org/Archive/2011/Conf/Schlag1314. pdf (last accessed 18
  April 2014).

\bibitem[Bergemann and Schlag, 2008]{bergemann2008pricing}
Bergemann, D. and Schlag, K.~H. (2008).
\newblock Pricing without priors.
\newblock {\em Journal of the European Economic Association}, 6(2-3):560--569.

\bibitem[Bulow and Klemperer, 1996]{bu1996}
Bulow, J. and Klemperer, P. (1996).
\newblock Auctions versus negotiations.
\newblock {\em The American Economic Review}, pages 180--194.

\bibitem[Bulow and Roberts, 1989]{BR89}
Bulow, J. and Roberts, J. (1989).
\newblock The simple economics of optimal auctions.
\newblock {\em Journal of Political Economy}, 97(5):1060--1090.

\bibitem[Cole and Roughgarden, 2014]{CR14}
Cole, R. and Roughgarden, T. (2014).
\newblock The sample complexity of revenue maximization.
\newblock In {\em Symposium on Theory of Computing, {STOC} 2014, New York, NY,
  USA, May 31 - June 03, 2014}, pages 243--252.

\bibitem[Devanur et~al., 2011]{DHKN11}
Devanur, N.~R., Hartline, J.~D., Karlin, A.~R., and Nguyen, C.~T. (2011).
\newblock Prior-independent multi-parameter mechanism design.
\newblock In {\em WINE}, pages 122--133.

\bibitem[Dhangwatnotai et~al., 2014]{dhangwatnotai2014revenue}
Dhangwatnotai, P., Roughgarden, T., and Yan, Q. (2014).
\newblock Revenue maximization with a single sample.
\newblock {\em Games and Economic Behavior}.

\bibitem[Fu et~al., 2014]{FHHK14}
Fu, H., Haghpanah, N., Hartline, J.~D., and Kleinberg, R. (2014).
\newblock Optimal auctions for correlated buyers with sampling.
\newblock In {\em {ACM} Conference on Economics and Computation, {EC} '14,
  Stanford, CA, USA, June 8-12, 2014}, pages 23--36.

\bibitem[Fu et~al., 2013]{FHH13}
Fu, H., Hartline, J.~D., and Hoy, D. (2013).
\newblock Prior-independent auctions for risk-averse agents.
\newblock In {\em ACM Conference on Electronic Commerce}, pages 471--488.

\bibitem[Huang et~al., 2014]{huang2014making}
Huang, Z., Mansour, Y., and Roughgarden, T. (2014).
\newblock Making the most of your samples.
\newblock {\em arXiv preprint arXiv:1407.2479}.

\bibitem[Myerson, 1981]{Mye81}
Myerson, R. (1981).
\newblock Optimal auction design.
\newblock {\em Mathematics of Operations Research}, 6(1):pp. 58--73.

\bibitem[Roughgarden et~al., 2012]{RTY12}
Roughgarden, T., Talgam-Cohen, I., and Yan, Q. (2012).
\newblock Supply-limiting mechanisms.
\newblock In {\em ACM Conference on Electronic Commerce}, pages 844--861.

\bibitem[Wilson, 1989]{Wilson1989}
Wilson, R.~B., . (1989).
\newblock Game-theoretic approaches to trading processes.
\newblock In {\em Advances in economic theory}. Cambridge University Press.

\end{thebibliography}
